\documentclass{article}

\usepackage[final]{neurips_2025}

\usepackage[utf8]{inputenc} 
\usepackage[T1]{fontenc}    
\usepackage{hyperref}       
\usepackage{url}            
\usepackage{booktabs}       
\usepackage{amsfonts}       
\usepackage{nicefrac}       
\usepackage{microtype}      
\usepackage{xcolor}         

\AtBeginDocument{%
	\providecommand\BibTeX{{%
			\normalfont B\kern-0.5em{\scshape i\kern-0.25em b}\kern-0.8em\TeX}}}
\usepackage{fancyhdr}

\usepackage{enumitem}
\usepackage{amsmath}
  
\usepackage{tcolorbox}
\usepackage{amssymb}
\usepackage{bm}
\usepackage{nicefrac}
\usepackage{booktabs}
\usepackage{array}
\usepackage{multirow}
\usepackage{threeparttable}
\usepackage{makecell}
\usepackage[procnumbered,ruled,vlined,linesnumbered]{algorithm2e}
\usepackage{siunitx}
\usepackage{stfloats}
\usepackage{graphicx}
\usepackage{subfigure}
\usepackage{hyperref}
\usepackage{graphicx}
\usepackage{caption}
\usepackage{subcaption}

\newtheorem{problem}{Problem}
\newtheorem{theorem}{Theorem}[section]

\newtheorem{lemma}[theorem]{Lemma}

\newenvironment{fminipage}%
{\begin{Sbox}\begin{minipage}}%
		{\end{minipage}\end{Sbox}\fbox{\TheSbox}}

\def\calG{\mathcal{G}}

\def\calR{\mathcal{R}}

\newcommand{\removelatexerror}{\let\@latex@error\@gobble}

\newcommand\LL{\bm{\mathit{L}}}

\newcommand\zz{\boldsymbol{\mathit{z}}}

\newcommand\bb{\boldsymbol{\mathit{b}}}
\newcommand\cc{\boldsymbol{\mathit{c}}}

\newcommand\ee{\boldsymbol{\mathit{e}}}

\newcommand\sss{\boldsymbol{\mathit{s}}}

\renewcommand\SS{\boldsymbol{\mathit{S}}}
\renewcommand\AA{\boldsymbol{\mathit{A}}}

\newcommand\DD{\boldsymbol{\mathit{D}}}

\newcommand\PP{\boldsymbol{\mathit{P}}}

\newcommand\II{\boldsymbol{\mathit{I}}}

\DeclareMathOperator*{\argmin}{arg\,min}
\DeclareMathOperator*{\argmax}{arg\,max}

\DontPrintSemicolon
\SetKw{KwAnd}{and}
\SetFuncSty{textsc}
\SetKwInOut{Input}{Input\ \ \ \ }

\SetKwInOut{Output}{Output}

\usepackage{tabularx}
\usepackage{stfloats}

\DontPrintSemicolon
\SetKw{KwAnd}{and}
\SetFuncSty{textsc}
\SetKwInOut{Input}{Input\ \ \ \ }
\SetKwInOut{Output}{Output}

\usepackage{tabularx}
\usepackage{stfloats}

\begin{document}
\title{Fast Computation and Optimization for  Opinion-Based Quantities of   Friedkin-Johnsen Model}

%

\author{%
  Haoxin Sun, Yubo Sun, Xiaotian Zhou, Zhongzhi Zhang\thanks{Corresponding author.} \\
  College of Computer Science and Artificial Intelligence \\
  Fudan University\\
  \texttt{\{23110240089,25110890019\}@m.fudan.edu.cn, \{20210240043,zhangzz\}@fudan.edu.cn} \\
}


\maketitle

\begin{abstract}
    In this paper, we address the problem of fast computation and optimization of opinion-based quantities in the Friedkin–Johnsen (FJ) model. We first introduce the concept of partial rooted forests, based on which we present an efficient algorithm for computing relevant quantities using this method. Furthermore, we study two optimization problems in the FJ model: the Opinion Minimization Problem and the Polarization and Disagreement Minimization Problem. For both problems, we propose fast algorithms based on partial rooted forest samplings. Our methods reduce the time complexity from linear to sublinear. Extensive experiments on real-world networks demonstrate that our algorithms are both accurate and efficient, outperforming state-of-the-art methods and scaling effectively to large-scale networks.
\end{abstract}

\section{Introduction}
 
Online social networks and social media have become integral to our daily lives, fundamentally altering how individuals communicate, exchange, and form opinions~\cite{Le20,KiArNiNi17,SiCiAr20,YaJo23,ZhBaZh21,SuZh23}. Recent studies suggest that, unlike traditional forms of communication, online interactions in the digital era have profoundly impacted human behavior, facilitating the widespread, critical, and complex propagation of information~\cite{NoCaFlMaRa22}. To better understand the mechanisms driving opinion formation and dissemination, various mathematical frameworks for opinion dynamics have been developed~\cite{JiMiFrBu15,PrTe17,DoZhKoDiLi18,AnYe19}. Among these models, the Friedkin-Johnsen (FJ) model~\cite{FrJo90} stands out as one of the most widely used, with applications spanning multiple fields~\cite{BeWaVaHoShAl21,FrPrTePa16}. For instance, a modified FJ model was recently used to study the Paris Agreement negotiations, revealing key factors behind the achieved consensus~\cite{BeWaVaHoShAl21}.

The fundamental concept in opinion dynamics is the opinion itself, which serves as the basis for many opinion-based quantities that have garnered significant attention. Among these quantities are the overall opinion, which reflects public sentiment on specific issues, and various social phenomena such as polarization, disagreement, and conflict. Given the importance of these opinion-based quantities, a key challenge is how to effectively compute and optimize them. A Laplacian solver-based approach for computing these quantities was proposed in~\cite{XuBaZh21, XuZhGuZhZh22}, followed by a sampling-based algorithm to accelerate the computation process~\cite{NeDoPe24}. For the optimization of opinion-based quantities, a range of methods has been  introduced, including matrix inversion~\cite{GiTeTs13}, vector projection~\cite{ZhZh21, ZhBaZh21}, eigencentralities~\cite{AmSi19}, convex optimization~\cite{MuMuTs18}, and sampling techniques~\cite{SuZh23}. Although many effective methods have been proposed for computing and optimizing opinion-based quantities in the FJ model, they are often limited to specific problems or suffer from high time complexity. Consequently, a unified framework to efficiently address both computation and optimization of opinion-based quantities in the FJ model is imperative.

In this paper, we propose several algorithms to compute and optimize the opinion-based quantities. Our main contributions are as follows:
\begin{itemize}[leftmargin=12pt,topsep=6pt]
    \item We introduce the concept of partial rooted forests and, based on which we propose a fast sampling-based algorithm for computing opinion-based quantities. Our methods effectively capture the essential structural information of the graph, ensuring both efficiency and accuracy.
    \item We address two optimization problems in the FJ model: minimizing a weighted average of expressed opinions, and reducing polarization and disagreement via edge addition. For both problems, we design fast algorithms based on partial rooted forest samplings, reducing the time complexity from linear to sublinear.
    \item We conduct extensive experiments on a variety of real-world networks, which shows that our algorithms demonstrate both high accuracy and efficiency compared to state-of-the-art methods and are scalable to large networks.
\end{itemize}

\section{Related Work}


Mathematical modeling plays a crucial role in understanding opinion dynamics, and numerous models have been proposed over the years to capture various aspects of opinion formation~\cite{JiMiFrBu15,PrTe17,DoZhKoDiLi18,AnYe19}. Among these, the Friedkin-Johnsen (FJ) model~\cite{FrJo90} stands out as a foundational model, building upon and extending the DeGroot model~\cite{De74}. Given its theoretical importance and practical applications, the FJ model has garnered significant attention since its introduction. A sufficient stability condition for the FJ model was derived in~\cite{RaFrTeIs15}, and the model's average innate opinion was characterized in~\cite{DaGoPaSa13}. Additionally, the vector of expressed opinions at equilibrium was formulated in~\cite{DaGoPaSa13,BiKlOr15}. Further interpretations and insights into the FJ model have also been provided~\cite{GhSr14,BiKlOr15}. 

The sum of opinions  has attracted significant attention, with various research groups addressing the optimization problem of maximizing overall opinion through leader selection~\cite{YiCaPa21,HuZa22,ZhZh23} or link recommendation~\cite{ZhZh21,ZhZhLiZh23} based on the DeGroot model. In the case of the FJ model, node-based strategies have been proposed over the past decade to optimize the sum of opinions on unsigned graphs, including modifications to initial opinions~\cite{AhDeHaMaYa15}, the expression of opinions~\cite{GiTeTs13,SuZh23}, and sensitivity to persuasion~\cite{AbKlPaTs18,ChLiSo19,AbChKlLiPaSoTs21,MaMiTaTa21}.  Our solution for optimizing the average expressed opinion in the FJ model is both more efficient and effective compared to the state-of-the-art approach presented in~\cite{SuZh23}.

The explosive growth of social media and online social networks has given rise to several social phenomena, including polarization~\cite{MaTeTs17,MuMuTs18,AmBoSi19}, disagreement~\cite{MuMuTs18}, filter bubbles~\cite{BaChLa23,La22}, conflicts~\cite{ChLiDe18}, and controversies~\cite{ChLiDe18}. In response to these challenges, research has evolved in various directions, with fast algorithms being developed to efficiently compute these quantities~\cite{XuBaZh21,XuZhGuZhZh22,NeDoPe24}. Some studies have focused on identifying groups of users with an open attitude toward opposing information~\cite{GeMiYoZe18,FaBaGe20}, aiming to connect users with differing opinions in order to mitigate the filter bubble effect~\cite{ToRoGo21,ZhBaZh21,AmSi19,MuMuTs18,ZhZh22}. In this paper, we study the problem of computing opinion-based quantities in the FJ model~\cite{NeDoPe24, XuBaZh21} and  two optimization problems in the FJ model~\cite{SuZh23,ZhBaZh21}. Our proposed algorithms demonstrated greater efficiency and effectiveness compared to the state-of-the-art methods in~\cite{NeDoPe24, SuZh23,ZhBaZh21}.

\section{Preliminaries}


\subsection{ Graph and  Laplacian Matrix}
Let $\calG=(V,E)$ denote  an unweighted simple undirected graph, which consists of $n=|V|$ nodes  and $m=|E|$ edges.  An edge \( (v_i,v_j) \in E \) indicates the edge between node \(v_i\) and node \(v_j\). In what follows, $v_i$ and $i$ are used interchangeably to represent node $v_i$ if incurring no confusion.  The structure of graph $\calG=(V,E)$ is captured by its adjacency matrix $\AA=(a_{ij})_{n\times n}$, where $a_{ij} = a_{ji} = 1$ if there is an edge between node $v_i$ and node $v_j$   and $ 0$ otherwise. The  degree $d_i$ of node $i$  is defined by \(d_i =\sum_{j=1}^n a_{ij}\). The  diagonal  degree matrix representing the degrees of graph $\calG$ is ${\DD} = {\rm diag}(d_1, d_2, \ldots, d_n)$, and the Laplacian matrix is ${\LL}={\DD}-{\AA}$. For any given node $i$, $N_i$  denotes the  set  of its neighbors, meaning $N_i =\{ j: (i,j)\in E\}$.   A path $P$ from node $v_1$ to $v_j$ is a sequence of alternating nodes and edges $v_1$,$(v_1,v_2)$,$v_2$,$\cdots$, $v_{j-1},(v_{j-1},v_j)$, $v_j$ where each node is unique and every edges connects $v_i$  to $v_{i+1}$. A loop  is a path plus an edge from the ending node to the starting node.

\subsection{Friedkin-Johnsen Model}
The Friedkin-Johnsen (FJ) model~\cite{FrJo90} is a widely used framework for modeling opinion evolution and formation. In the FJ model applied to a graph $\calG = (V, E)$, each node (or agent) $i \in V$ is characterized by two types of opinions: an internal opinion $s_i$ and an expressed opinion $z_i(t)$ at time $t$. The internal opinion $s_i \in [0, 1]$ reflects the inherent stance of node $i$ on a particular topic. Throughout the opinion evolution process, the internal opinion $s_i$ remains fixed, while the expressed opinion $z_i(t)$ evolves at time $ t+1$ as  $z_i(t+1) = {(s_i +\sum_{j\in N_i} a_{ij}z_j(t))}/{(1+\sum_{j\in N_i} a_{ij})}.$


 Let $ \sss = (s_1,s_2,\cdots,s_n)^\top$ denote the vector of internal opinions, and let $ \zz(t) = (z_1(t),z_2(t),\cdots,z_n(t))^\top $ denote the vector of expressed opinions at time $ t $. According to~\cite{BiKlOr15}, as $ t $ approaches infinity, $ \zz(t) $ converges to an equilibrium vector  $ \zz = (z_1,z_2,\cdots,z_n)^\top$ satisfying $ \zz = (\II+\LL)^{-1}\sss $. Define $\mathbf{\Omega}=\left(\II+\LL\right)^{-1}=(\omega_{ij})_{n \times n}$, referred to as the forest matrix~\cite{ChSh97,ChSh98}. The forest matrix $\Omega$ is doubly stochastic for undirected graphs, with all its components in the interval $[0,1]$. Furthermore, for each column, the diagonal elements are greater than the off-diagonal elements, that is  $ 0\leq\omega_{ji}< \omega_{ii}\leq 1$ for any pair of different nodes $ i$ and $j $, and the diagonal element $\omega_{ii}$ of matrix $\mathbf{\Omega}$ satisfies $\frac{1}{1+d_i}\leq \omega_{ii} \leq \frac{2}{2+d_i}$~\cite{SuZh23}. The forest matrix $\mathbf{\Omega}$ serves as the fundamental matrix of the FJ model for opinion dynamics~\cite{GiTeTs13}. For every node $i \in V$, its expressed opinion $z_i$ is given by $z_i=\sum^n_{j=1}  \omega_{ij}s_j$, a convex combination of the internal opinions for all nodes. 

\section{Partial Rooted Forest Samplings for Estimating the Forest Matrix}
The forest matrix $\mathbf{\Omega}$ is the fundamental matrix of the FJ model and plays a key role in its computation and optimization. Existing methods~\cite{sun2024efficient,sun2024fast,SuZh23} rely on generating rooted spanning forests over the entire graph, requiring $O(ln)$ time.   We propose a local approach, called partial rooted forest samplings, as a natural extension of previous methods by using absorbing random walks. It avoids full-graph traversal and offers improved efficiency and scalability.

\subsection{Absorbing Random Walk }\label{sec-absorbing}
The forest matrix $\mathbf{\Omega}$  is  doubly stochastic for undirected graphs. We initiate a random walk from node $i$ and assume the walk is currently at node $k \in V$. At node $k$, the walk has a probability of $\frac{1}{1 + d_k}$ to stop, and with probability $1 - \frac{1}{1 + d_k}$, it moves to a randomly selected neighbor of node $k$. If the walk eventually stops at node $q$, we say that the walk has been absorbed by $q$. In this context, $\omega_{ij}$ represents  the probability that a walk starting at node $i$ will eventually be absorbed at node  $j$.

To formally explain this, we expand the forest matrix as an infinite series. Define $\PP =(\II+\DD)^{-1}\AA$, then  the forest matrix has the following form: $  \mathbf{\Omega}= (\II -(\II+\DD)^{-1}\AA)^{-1}(\II+\DD)^{-1} = (\II-\PP)^{-1}(\II+\DD)^{-1} = \sum_{k\geq 0}\PP^k(\II+\DD)^{-1}$. Here, the $i,j$-th entry of $\PP^k$  represents the probability that a walk starting at $i$ takes exactly $ k$ steps to reach $j$. Multiplying this by $\frac{1}{1 + d_j}$ gives the probability that the walk takes $k$ steps and is then absorbed at $j$.Then $\omega_{ij} $  can be expressed as $        \omega_{ij} = \ee^\top_{i}\mathbf{\Omega}\ee_j= \sum_{k\geq 0}\ee^\top_{i}\PP^k(\II+\DD)^{-1}\ee_j =\frac{1}{1+d_j} \sum_{k\geq 0}\ee^\top_{i}\PP^k\ee_j.$
Choose an initial node $i \in V$, and perform the absorbing random walk several times. Using the probabilistic interpretation of the forest matrix, where $\omega_{ij}$ represents the probability of a random walk starting at node $i$ and being absorbed at node $j$, we can estimate the $i$-th row of the forest matrix. The following lemma establishes the expected time complexity of the absorbing random walk:

 \begin{lemma}\label{le-absorbedrw}
    For an undirected graph $\calG = (V, E)$ and its related forest matrix $\mathbf{\Omega}$, the expected length of an absorbing random walk starting at node $i$ is $\sum_{j=1}^n{\omega_{ij}}d_j$. If the initial node $i$ is chosen randomly, the expected length of the random walk to absorption is $\bar{d} = \frac{1}{n}\sum_{i=1}^nd_i = \frac{2m}{n}$.
 \end{lemma}

 \begin{proof}
     Let $l_i$ denote the expected length of an absorbing random walk starting at node $i$, and define the vector ${\boldsymbol{\mathit{l}}} = (l_1,\cdots,l_n)^\top$. For any $i \in V$, the relationship between $l_i$ and its neighbors is given by
     $ l_i = \frac{d_i}{1+d_i}(1+\frac{1}{d_i}
         \sum_{j\in N_i^+}l_j)$. Rewriting this in matrix form, we have
$         {\boldsymbol{\mathit{l}}} =(\II+\DD)^{-1}\DD \boldsymbol{1} +(\II+\DD)^{-1}\AA{\boldsymbol{\mathit{l}}}$. 
     Solving this equation yields ${\boldsymbol{\mathit{l}}} = \mathbf{\Omega}\DD \boldsymbol{1}$. Thus, for each $i \in V$, the expected length is $l_i = \sum_{j=1}^n{\omega_{ij}}d_j$. If the initial node $i$ is chosen uniformly at random, the expected length of the random walk is $                  \bar{l} = \frac{1}{n}\sum_{i=1}^nl_i =\frac{1}{n}\sum_{i=1}^n\sum_{j=1}^n{\omega_{ij}}d_j =
                  \frac{1}{n}\sum_{j=1}^n(\sum_{i=1}^n{\omega_{ij}})d_j = \bar{d},$ where the property   $\sum_{i=1}^n \omega_{ij} = 1$ in undirected graphs is used. 
 \end{proof}

By Lemma~\ref{le-absorbedrw}, we can estimate $\zz$ by performing absorbing random walks from each node in $V$ multiple times. However, this method requires iterating over all nodes, which is computationally inefficient.

 \subsection{Partial Rooted Forest Sampling }
In this subsection, we introduce a sampling method for constructing a partial rooted forest, which significantly reduces sampling time. We first define the concept of a partial rooted forest. In an undirected graph $\calG=(V,E)$, a  rooted tree of  $\calG$ is a connected subgraph without cycle, where one node is set to be the root. An isolated node is considered as a tree with the root being itself. A partial rooted forest $\phi = (V_{\phi},E_{\phi})$ is a subgraph of $G$, where all connected components are rooted trees. The root set $\mathcal{R}(\phi)$ of $\phi$ is defined as the collection of root nodes from all rooted trees in $\phi$. Since each node $i$ in $V_\phi$  belongs to  a specific rooted tree, we define a function $r_{\phi}(i): V_{\phi} \rightarrow \calR(\phi) $ mapping node $i$ to the root of its corresponding rooted tree.  If $r_{\phi}(i) = j$, then $j \in \mathcal{R}(\phi)$, and both $i$ and $j$ belong to the same rooted tree in $\phi$.

Next, we describe the procedure for generating a partial rooted forest using a loop-erased absorbing random walk. This involves the loop-erasure technique, which eliminates loops from a random walk in chronological order~\cite{LaFr79, La80}. By applying this technique, we can utilize the paths generated by the absorbing random walk instead of discarding them. The procedure is as follows:

(i) \textbf{Initialization:}
Let $\SS = \{v_1, \ldots, v_p\} \subset V$ be a set of $p \geq 2$ nodes. Initialize $\phi = (V_{\phi}, E_{\phi}) = (\emptyset, \emptyset)$ and set the indicator $i = 1$.

(ii)\textbf{ Performing absorbing Random Walk:}
Select the first node $u = v_i$. Perform an absorbing random walk starting from $u$. Suppose that the walk is currently at a node $k \notin V_{\phi}$, and the walk is absorbed with probability $\frac{1}{1 + d_k}$. If the node is absorbed at $k$, $k$ is marked as a root node. Otherwise, with probability $1 - \frac{1}{1 + d_k}$, the walk moves to a uniformly chosen neighbor of $k$. If the walk reaches a node $k \in V_{\phi}$, it is immediately absorbed.

(iii)\textbf{ Loop Erasure:}
Let $P_u$ represent the trajectory of the walk  from node $u$ to its absorption point. Apply the loop-erasure technique to $P_u$ to obtain a simple path $\hat{P}_u$. Add the nodes and edges from $\hat{P}_u$ to $V_{\phi}$ and $E_{\phi}$, respectively.

(iv)\textbf{ Iteration:}
If $i < p$, increment $i$ by 1 and repeat step (ii). Otherwise, terminate the process and return the partial rooted forest $\phi$.

The procedure is detailed in Algorithm PFS (\underline{P}artial Rooted \underline{F}orest \underline{S}ampling); due to space constraints, the pseudocode is provided in the appendix. 

The following theorem establishes the expected time complexity of the Partial Rooted Forest Sampling (PFS) algorithm. The time complexity is $O(rpl)$, where $r$ is a ratio dependent on the graph structure, with an upper bound   $\bar{d}$. In real-world web and social networks, the average degree is typically $O(\log n)$ or a constant~\cite{WaYaXi17, XuYiZhKaZh20}. Furthermore, our experimental results demonstrate that $r$ is consistently smaller than $\bar{d}$. These observations collectively indicate the efficiency of our algorithm in practical scenarios.

\begin{theorem}\label{th-partial-rooted-forest}
For an undirected graph $\calG = (V, E)$ and a set of $p$ nodes $\SS \subset V$, the expected time complexity of Algorithm \textsc{PFS} is $O(rpl)$, where $r$ is the ratio of the expected number of nodes in a partial rooted forest $\phi \in L$ to $p$, and $l$ is the number of samples. $r$ satisfies relation $1\leq r \leq \bar{d}$ if we randomly sample the set $\SS$ from $V$.  Algorithm \textsc{PFS}  is more efficient than performing absorbing random walks individually for each node in $\SS$.
\end{theorem}

\begin{proof}
	The algorithm repeats the sampling process $l$ times and outputs a list of partial rooted forests. To establish the expected time complexity, consider a single partial rooted forest $\phi$ in the list $L$. Let $q  = rp$ represent the expected number of nodes in $\phi$. We will show that the expected time complexity for generating $\phi$ is $O(q)$.
	
	First, we extend the partial rooted forest $\phi$ to a complete spanning rooted forest $\phi'$ by executing the procedure described in lines 4–16 of Algorithm~\ref{alg:partial-rooted-forest} for the remaining nodes in $V \setminus V_\phi$. According to~\cite{SuZh23}, the expected time complexity of sampling the complete rooted forest $\phi'$ is $O(n)$. Furthermore, Wilson~\cite{Wi96} demonstrates that the order of node processing does not affect the generation of a spanning rooted forest. Hence, the partial rooted forest $\phi$ can be viewed as a subgraph of one complete rooted forest $\phi'$.

	Next, we analyze the expected time complexity of generating the partial rooted forest $\phi$. Marchel~\cite{Ma00} shows that the expected absorption time for each node in a rooted forest is $\omega_{ii}(1 + d_i)$. Sun~\cite{SuZh23} refines this result, proving that $\omega_{ii}(1 + d_i) \leq \frac{2(1 + d_i)}{2 + d_i} \leq 2$, which is $O(1)$. 

	Therefore, the expected time complexity for sampling a single partial rooted forest $\phi$ is $O(q)$, leading to a total expected time complexity of $O(ql)$ for the algorithm.

	To give a rough estimation of $r$, if we sample $\SS$ randomly from $V$, an upper bound of $r$ is $O( \bar{d})$ according to Lemma~\ref{le-absorbedrw}.  Moreover, Algorithm~\ref{alg:partial-rooted-forest} retains only the necessary branches in the forest for each node. The sampling process terminates as soon as the walk reaches a node already in $V_\phi$, making it more efficient than performing absorbing random walks for all nodes in $\SS$ individually, as described in Subsection~\ref{sec-absorbing}. 
\end{proof}

\section{Fast Sampling Method  for Opinion-Based Quantities in FJ  Model}

\subsection{Definitions of  Opinion-Based Quantities}
The FJ model includes several important quantities for analyzing the properties of social groups. Below, we provide a concise overview of these measures, following prior works~\cite{MuMuTs18,MaTeTs17,XuBaZh21}. 

Consider an undirected graph  $\calG = (V,E)$. 
 Disagreement $D(\calG)$ is a measure of the difference in opinions between neighbors, calculated as $D(\calG)  = \sum_{(i,j)\in E} (z_i-z_j)^2 = \zz^\top\LL\zz$.  Polarization  $P(\calG)$ reflects how far individual opinions deviate from the average-expressed opinion. Using the vector ${\boldsymbol{\mathit{\bar{z}}}} = \zz - \frac{\zz^\top\boldsymbol{1}}{n}\boldsymbol{1}$, polarization is defined as $ P(\calG)  = \sum_{i\in V} (z_i-\bar{z})^2 = {\boldsymbol{\mathit{\bar{z}}}} ^\top{\boldsymbol{\mathit{\bar{z}}}}$. 
Internal conflict $I(\calG)$ quantifies the discrepancy between the expressed opinions $\zz$ and the initial opinions $\sss$, given by $I(\calG)  = \sum_{i\in V} (z_i-s_i)^2 = \zz^\top\LL^2\zz$. Controversy $C(\calG)$ represents the overall magnitude of expressed opinions $ C(\calG)  = \sum_{i\in V} z_i ^2 = \zz^\top \zz$. 
 Disagreement-controversy $DC(\calG)$  is a combined measure capturing both disagreement and controversy $DC(\calG)  = D(\calG)+C(\calG)=\sum_{i\in V} s_iz_i = \sss^\top \zz$.

\subsection{Estimating Opinion-Based  Quantities}

In this subsection, we present a sampling-based algorithm for QE (Opinion-Based \underline{Q}uantities \underline{E}stimation). The algorithm leverages the partial rooted forest samplings method to estimate the opinion-based quantities efficiently. We first provide a lemma to show that the map between the nodes in  $\SS$ to its root nodes in the partial rooted forest can be used to estimate the opinion-based quantities.
\begin{lemma}\label{le-estimate}
	For an undirected graph $\calG = (V, E)$, a set of  nodes $\SS \subset V$, and internal opinion vectors $\sss = (s_1,\cdots,s_n)^\top$, let $\phi$ be the partial rooted forest generated by Algorithm \textsc{PFS} . For node $i\in \SS$, let $r_{\phi}(i)$ denote the root node of $i$ in $\phi$. Then, the estimator $\widehat{z}_i = \sum_{j=1}^n \mathbb{I}_{\{r_{\phi}(i)=j\}}s_j$ is an unbiased estimator of $z_i$.
\end{lemma}

\begin{algorithm}[h!]
\caption{ $\textsc{PF-QE}(\calG,\SS,\sss,L)$}
\label{alg:estimation}
\Input{Graph $\calG=(V,E)$, node set $\SS = \{v_1, \ldots, v_p\}$, internal opinion   $\sss = (s_1,\cdots,s_n)^\top$} 
\Output{Estimates of opinion-based quantities}
Sample a list $L$  of $l$ partial rooted forests  \;
Initialize $\widehat{\zz}_{\SS} \leftarrow \mathbf{0}$ (a $p$-dimensional vector for storing $\widehat{z}_i$ for each $v_i \in \SS$)\;

\For{$k \leftarrow 1$ \KwTo $l$}{
	$\phi \leftarrow L[k]$\;
	\For{$i \leftarrow 1$ \KwTo $p$}{
		$r_i \leftarrow r_{\phi}(v_i)$, \ 
		$\widehat{\zz}_{\SS}[i] \leftarrow \widehat{\zz}_{\SS}[i] + s_{r_i}$\;
	}
}
$\widehat{\zz}_{\SS} \leftarrow \frac{\widehat{\zz}_{\SS}}{l}$,  $\widehat{C} \leftarrow \frac{n}{p}\sum_{i=1}^p \widehat{\zz}_{\SS}[i]^2$,  $\widehat{DC} \leftarrow \frac{n}{p}\sum_{i=1}^p s_i\widehat{\zz}_{\SS}[i]$, $\widehat{D} \leftarrow \widehat{DC} - \widehat{C}$\\ 
$\widehat{\bar{z}}_{\SS} = \frac{1}{p}\sum_{i=1}^{p}\widehat{\zz}_{\SS}[i]$, $\widehat{P} \leftarrow \frac{n}{p}\sum_{i=1}^p (\widehat{\zz}_{\SS}[i] - \widehat{\bar{ z}}_{\SS})^2$, $\widehat{I} \leftarrow \frac{n}{p}\sum_{i=1}^p (\widehat{\zz}_{\SS}[i] - s_i)^2$\;
\Return $ \widehat{D}, \widehat{P}, \widehat{I}, \widehat{C}, \widehat{DC}$\;
\end{algorithm}   

Lemma~\ref{le-estimate} shows that the estimator $\widehat{z}_i$ is unbiased for $z_i$. By applying the partial rooted forest samplings method, we can estimate the opinion-based quantities in the FJ opinion dynamics model, and we now give the pseudocode for the sampling-based algorithm.

Algorithm~\ref{alg:estimation} PF-QE (\underline{P}artial Rooted \underline{F}orest Method for QE ) provides the pseudocode for estimating the opinion-based quantities in the FJ opinion dynamics model. The algorithm leverages the partial rooted forest samplings method to estimate the opinion-based quantities efficiently. Instead of estimating all nodes in $V$, algorithm~\ref{alg:estimation} focuses on a set of nodes $\SS$ with $p$ nodes, and samples $l$ partial rooted forests. Now we provide a theorem,  showing the parameters $p$ and $l$ for a given error tolerance.  

\begin{theorem}\label{th-estimation}
	For an undirected graph $\calG = (V, E)$ and internal opinion vectors $\sss = (s_1,\cdots,s_n)^\top$, if we choose $p = O(\frac{1}{\epsilon^2}\log\frac{1}{\delta})$, $l = O(\frac{1}{\epsilon^2}\log\frac{1}{\delta})$, and node set $\SS$ is randomly sampled from $V$, then Algorithm~\ref{alg:estimation}   estimates the opinion-based quantities $D(\calG)$, $P(\calG)$, $I(\calG)$, $C(\calG)$, and $DC(\calG)$ within an absolute error of $n\epsilon$ with probability at least $1-\delta$. The total time complexity of Algorithm~\ref{alg:estimation} is $O(rpl)$.
\end{theorem}

\section{Partial Forest Sampling Techniques for Optimization Problems in FJ model}

\subsection{ Opinion Minimization Problem}

\begin{tcolorbox}[boxrule={1pt}]
	\begin{problem}\label{Pr-opt}[\underline{{Op}}inion \underline{{Min}}imization Problem (OpMin)]
		Given an undirected graph $ \calG = (V,E) $, a parameter vector $ \cc = (c_1,\cdots, c_n)^\top$, where $c_i \in [0,1]$, an integer $k\ll n$, and an iternal opinion vecter $\sss$, we aim to find the set $ H \subseteq V $ of $k$ nodes,  and set their internal opinions to $0$,  so that the function  $f(\cc,H) = \frac{1}{n} \cc^\top\zz = \frac{1}{n}\sum_{i=1}^n c_iz_i$ is minimized. That is,
			\begin{equation}\label{EOMi}
				H =  \argmin_{U \subseteq V, |U|= k} f(\cc,U).
			\end{equation}
	\end{problem}
\end{tcolorbox}
Our formulation of Problem~\ref{Pr-opt} generalizes the setting in~\cite{SuZh23}. Specifically, when $c_i = 1$ for all $i \in V$, it reduces to the average-expressed opinion minimization problem studied in~\cite{SuZh23}. To address this problem efficiently, we propose a fast algorithm based on partial rooted forest samplings, which reduces the time complexity from $O(ln)$ to $O(rpl)$.

Recall that $\zz = (\II+\LL)^{-1}\sss$ is the equilibrium vector of the expressed opinions. Then $z_i = \sum_{j=1}^n \omega_{ij}s_j$. We can rewrite the objective function as $f(\cc,H) = \frac{1}{n}\sum_{i=1}^n c_iz_i = \frac{1}{n}\sum_{i=1}^n c_i\sum_{j=1}^n \omega_{ij}s_j = \frac{1}{n}\sum_{j=1}^n (\sum_{i=1}^n c_i\omega_{ij} )s_j $. Define $\gamma_j= \frac{1}{n}\sum_{i=1}^n c_i\omega_{ij}$, then the objective function can be rewritten as $f(\cc,H) = \sum_{j=1}^n \gamma_js_j$. In this case, the objective function is linear with respect to the internal opinions $s_i$. To solve OpMin, we need to find the set $H$ of $k$ nodes with the smallest $\gamma_js_j$ values.  

In~\cite{SuZh23}, the authors use  the forest sampling method to solve the  problem. However, the algorithm needs to perform a random walk process across all nodes in the graph, which is computationally expensive.  To overcome this, we propose an efficient method using the partial rooted forest samplings. The key to solving OpMin lies in estimating   the $\gamma_j = \frac{1}{n}\sum_{i=1}^n c_i\omega_{ij}$  for all nodes in the graph. Suppose that we choose a set of $p$ nodes $\SS$ randomly from $V$, and sample a partial rooted forest $\phi$. Then, define the estimator $\widehat{\gamma}_j(\phi) = \frac{1}{p}\sum_{i\in \SS} c_i\mathbb{I}_{\{r_{\phi}(i)=j\}}$.  We sample $l$ partial rooted forests $\phi_1,\cdots,\phi_l$, and estimate the $\gamma_j$ values using $\widehat{\gamma}_j = \frac{1}{l}\sum_{k=1}^l \widehat{\gamma}_j(\phi_k)$.  We detail the algorithm for solving OpMin in Algorithm~\ref{alg:optimization}.
 
\begin{algorithm}[htbp!]
\caption{$\textsc{PF-OpMin}(\calG,\SS,\cc,l,k)$}
\label{alg:optimization}
\Input{Graph $\calG=(V,E)$, node set $\SS = \{v_1, \ldots, v_p\}$, parameter vector $\cc = (c_1,\cdots,c_n)^\top$, number of forests $l$, number of nodes $k$}
\Output{Optimal set $\widehat{H}$}
Initialize $\widehat{\gamma} \leftarrow \mathbf{0}$, $\widehat{H} \leftarrow \emptyset$.   Sample a list $L$  of $l$ partial rooted forests  \;
\For{$t \leftarrow 1$ \KwTo $l$}{
	$\phi \leftarrow$ $L[t]$\;
	\For{$i \leftarrow 1$ \KwTo $p$}{
		$j \leftarrow r_{\phi}(v_i)$, \quad	$\widehat{\gamma}[j] \leftarrow \widehat{\gamma}[j] + c_{v_i}/lp$\;
	}
}
\For{$i \leftarrow 1$ \KwTo $k$}{
	$j \leftarrow \argmax_{v \in V \setminus H} \widehat{\gamma}[v] s_v$, \quad $\widehat{H} \leftarrow \widehat{H}\cup \{j\}$\;
}
\Return $\widehat{H}$\;
\end{algorithm}   

Algorithm~\ref{alg:optimization} \textsc{PF-OpMin}(\underline{P}artial Rooted \underline{F}orest Method for \underline{OpMin}) provides the pseudocode for solving OpMin. The algorithm leverages the partial rooted forest samplings method to estimate the $\gamma_j$ values efficiently. By sampling $l$ partial rooted forests, we estimate the $\gamma_j$ values using the estimator $\widehat{\gamma}_j$. The algorithm then selects the $k$ nodes with the smallest $\widehat{\gamma}_js_j$ values as the optimal set $\widehat{H}$. 

Consider an undirected graph $\calG = (V, E)$ and a parameter vector $\cc = (c_1, \cdots, c_n)^\top$, where each $c_i$ is a constant. By setting the parameters $p = O\left(\frac{1}{\epsilon^2}\log\frac{1}{\delta}\right)$ and $l = O\left(\frac{1}{\epsilon^2}\log\frac{1}{\delta}\right)$, and selecting the node set $\SS$ through random sampling from $V$, Algorithm~\ref{alg:optimization} can solve OpMin within an absolute error of $k\epsilon$, with probability at least $1 - \delta$. Specifically, the approximation satisfies $f(\cc, H) - f(\cc, \widehat{H}) \leq k\epsilon$. This result follows a proof strategy analogous to that of Theorem~\ref{th-estimation} and Theorem 5.5 in~\cite{SuZh23},  where $r$ is the ratio of the expected number of nodes in a partial rooted forest $\phi \in L$ to $p$. Our algorithm runs in $O(rpl)$ time and outperforms the state-of-the-art method~\cite{SuZh23}, which requires $O(ln)$ time. Since $p \ll n$ and, as shown in the experimental section, the empirical results on real-world networks indicate that $r < 15$, our approach improves the time complexity from linear to sublinear.


\subsection{Polarization and Disagreement Minimization problem}
In this subsection, we consider the Polarization and Disagreement Minimization Problem proposed in~\cite{ZhBaZh21}. The problem focuses on selecting a set of edges not present in the original graph to minimize the sum of polarization and disagreement. We define the P-D index as $\mathcal{I}(\calG) = D(\calG) + P(\calG) = \bar{\sss}^\top \mathbf{\Omega} \bar{\sss}$, and denote the augmented graph as $\calG + T = (V, E \cup T)$. The problem is formally described as follows:

\begin{tcolorbox}[boxrule={1pt}]
	\begin{problem}\label{Pr-pd}[\underline{P}olarization and \underline{D}isagreement \underline{Min}imization Problem (PDMin)]
		Given an undirected graph $ \calG = (V, E) $ and a candidate edge set $ E_C \subseteq \binom{V}{2} \setminus E $, the goal is to select a subset $ T \subseteq E_C $ of $k \ll n$ edges to add to the graph, such that the P-D index of the new graph is minimized. Define the objective function as $f(T) = \mathcal{I}(\calG) - \mathcal{I}(\calG + T)$, and find
		\begin{equation}\label{PDMP}
			T = \argmax_{T \subseteq E_C,\ |T| = k} f(T).
		\end{equation}
	\end{problem}
\end{tcolorbox}
In~\cite{ZhBaZh21}, the authors show that the problem is combinatorial in nature and propose a greedy algorithm to solve it. They prove that for a candidate edge $e \in E_C$ connecting nodes $u$ and $v$, with   vector $\bb_e = \ee_u - \ee_v$, the marginal gain is given by $f(e) = \frac{\sss^\top \mathbf{\Omega} \bb_e \bb_e^\top \mathbf{\Omega} \sss}{1 + \bb_e^\top \mathbf{\Omega} \bb_e} = \frac{(z_u - z_v)^2}{1 + r_{uv}} \geq 0$, where $r_{uv}$ is the forest distance between nodes $u$ and $v$, defined as $ r_{uv} \triangleq \bb_{uv}^\top \mathbf{\Omega} \bb_{uv}$. A greedy algorithm computes the marginal gain for each candidate edge in $E_C$, selects the one with the largest gain, and repeats this process $k$ times to obtain the final set $T$.

We set the sample node set for partial rooted forests to be the set of all nodes that appear in the candidate edge set $E_C$, that is,   $\SS = \bigcup_{(i, j) \in E_C} \{i, j\}$. Then, we generate $l$ partial rooted forests $\{\phi_1, \cdots, \phi_l\}$ and use them to estimate the marginal gain $f(e)$ for each candidate edge $e \in E_C$. The detailed procedure is presented in Algorithm \textsc{PF-PDMin}(\underline{P}artial Rooted \underline{F}orest Method for \underline{PDMin}).  For any candidate edge $e$, if we set $l = O\left(\frac{1}{\epsilon^2} \log \frac{1}{\delta}\right)$,   $|\widehat{f}(e) - f(e)| < \epsilon$ is guaranteed to hold with probability at least $1 - \delta$, by applying Hoeffding's inequality~\cite{Ho94}.  In~\cite{ZhBaZh21}, the authors use the Johnson–Lindenstrauss Lemma~\cite{JoLi84} and a fast Laplacian solver~\cite{SpTe14} to approximate the marginal gain, resulting in an overall time complexity of $\widetilde{O}(mk)$. In contrast, we leverage partial rooted forest samplings to reduce the time complexity to $O(r|E_C|kl)$.


\begin{algorithm}[t!]
\caption{\textsc{PF-PDMin}($\calG, E_C, \sss, l, k$)}
\label{alg:mpdp}
\Input{Graph $\calG = (V, E)$, candidate edge set $E_C$, internal opinion vector $\sss = (s_1,\cdots,s_n)^\top$, number of forests $l$, number of edges $k$}
\Output{Selected edge set $\widehat{T}$}
Initialize $\widehat{T} \leftarrow \emptyset$,    $\SS \leftarrow \bigcup_{(i,j) \in E_C} \{i, j\}$ \;
\For{$t \leftarrow 1$ \KwTo $k$}{
    Sample $l$ partial rooted forests $\{\phi_1, \ldots, \phi_l\}$ on node set $\SS$ \;
    \ForEach{$(u,v) \in E_C \setminus \widehat{T}$}{
        Initialize $\widehat{z}_u \leftarrow 0$, $\widehat{z}_v \leftarrow 0$, $\widehat{r}_{uv} \leftarrow 0$ \;
        \For{$i \leftarrow 1$ \KwTo $l$}{
            $r_u \leftarrow r_{\phi_i}(u)$,\quad $r_v \leftarrow r_{\phi_i}(v)$, \quad
            $\widehat{z}_u \leftarrow \widehat{z}_u + s_{r_u}$,\quad $\widehat{z}_v \leftarrow \widehat{z}_v + s_{r_v}$ \;
            $\widehat{r}_{uv} \leftarrow \widehat{r}_{uv} + \mathbb{I}_{\{u = r_u\}} + \mathbb{I}_{\{v = r_v\}} - \mathbb{I}_{\{u = r_v\}} - \mathbb{I}_{\{v = r_u\}}$ \;
        }
        $\widehat{z}_u \leftarrow \widehat{z}_u / l$,\quad $\widehat{z}_v \leftarrow \widehat{z}_v / l$,\quad $\widehat{r}_{uv} \leftarrow \widehat{r}_{uv} / l$, \quad
        $\widehat{f}(u,v) \leftarrow {(\widehat{z}_u - \widehat{z}_v)^2} /({1 + \widehat{r}_{uv}})$ \;
    }
    $(u^*,v^*) \leftarrow \argmax_{(u,v) \in E_C \setminus T} \widehat{f}(u,v)$ \;
    $\widehat{T} \leftarrow \widehat{T} \cup \{(u^*,v^*)\}$,\quad $E \leftarrow E \cup \{(u^*,v^*)\}$, \quad $V \leftarrow V \cup \{u^*,v^*\} $ \;
}
\Return $T$ \;
\end{algorithm}

\section{Experiments}
In this section, we conduct extensive experiments on various real-life networks in order to evaluate the performance of our algorithms, in terms of accuracy and efficiency.  Our source code is publicly available on \url{https://github.com/HaoxinSun98/FJ-PF}.

 \textbf{Datasets and Equipment.}
 The datasets for chosen real networks are  accessed publicly through KONECT~\cite{Ku13} and SNAP~\cite{LeSo16}. Our experiments cover a varied selection of real-world networks, and the specifics of these datasets are outlined in Table~\ref{datasets}.  All experiments are carried out using the Julia programming language within a computational environment equipped with a 2.5 GHz Intel E5-2682v4 CPU and 256GB of primary memory.
 \begin{table*}[htbp!]\fontsize{8}{11}\setlength{\tabcolsep}{10pt} 
\caption{Statistics of the datasets used in the experiments, including the number of nodes, number of edges, average degree, and the parameter $r$, which denotes the ratio of the average number of nodes in a partial rooted forest.}\label{datasets}\centering 
	\begin{tabular}{ccccc}
	\hline
	\multirow{2}{*}{Network} & \multirow{2}{*}{Nodes} & \multirow{2}{*}{Edges} & \multirow{2}{*}{ $\bar{d}$} & \multirow{2}{*}{$r$} \\
							 &                        &                        &                            &                      \\ \hline
	Delicious                & 536,108                & 1,365,961              & 5.1                        & 2.4                  \\
	Youtube                  & 1,134,890              & 2,987,624              & 5.3                        & 2.4                  \\
	Pokec                    & 1,632,803              & 22,301,964             & 27.3                       & 9.0                 \\
	Orkut                    & 3,072,441              & 117,184,899            & 76.3                       & 13.6                \\
	Livejournal              & 7,489,073              & 112,305,407            & 30.0                       & 3.3                \\
	Twitter                  & 41,652,230             & 1,202,513,046          & 57.7                       & 7.7                  \\ \hline
	\end{tabular}
\end{table*}


\textbf{Baselines.}
For the FJ quantities estimation problem, we compare our proposed algorithms with the LapSolver~\cite{XuBaZh21} and LazyWalk~\cite{NeDoPe24}. For OpMin, we compare our algorithm with the forest sampling method \textsc{Fast}~\cite{SuZh23}. For PDMin, we compare our proposed algorithms with the greedy algorithm \textsc{FastGreedy}~\cite{ZhBaZh21}.

\begin{table*}[htbp!]\caption{Time complexity of our algorithms and baselines for three problems with parameter settings.}
\renewcommand{\arraystretch}{1.3}
\setlength{\tabcolsep}{2.2pt} 
\centering
\small
\begin{tabular}{cccccccccc}
\hline
\multicolumn{2}{c}{QE}                &  & \multicolumn{2}{c}{OpMin}                &  & \multicolumn{2}{c}{PDmin}                   &  & Settings                          \\ \cline{1-2} \cline{4-5} \cline{7-10} 
$\textsc{PF-QE}$ & $O(rpl)$           &  & $\textsc{PF-OpMin}$ & $O(rpl)$           &  & $\textsc{PF-PDmin}$   & $O(rkl|E_C|)$       &  & $r < 15$ , $l = 10^3$, $p = 10^4$ \\
LazyWalk         & $O(r'pl')$         &  & $\textsc{Fast}$     & $O(ln)$            &  & $\textsc{FastGreedy}$ & $\widetilde{O}(mk)$ &  & $r' = 4\times10^3$, $l' = 600$    \\
LapSolver        & $\widetilde{O}(m)$ &  & $\textsc{Exact}$    & $\widetilde{O}(m)$ &  & ---                   & ---                 &  & $k=50$, $ |E_C| = 10^4$           \\ \hline
\end{tabular}
\end{table*}

 \subsection{ Estimation for Relevant Quantities}
In this subsection, we evaluate the performance of our proposed algorithms for estimating relevant quantities for FJ model on real-world networks. We compare our algorithms with LapSolver~\cite{XuBaZh21} and LazyWalk~\cite{NeDoPe24}. LapSolver achieves high accuracy~\cite{XuBaZh21} with a $10^{-6}$ relative error and is treated as the ground truth. Following the setting in~\cite{NeDoPe24}, we set the number of samples to $p = 10,000$ and vary the sampling parameters to compare the running time and mean relative error of the expressed opinions computed by PF-QE and LazyWalk. The results are shown in Figure~\ref{f1}. Both PF-QE and LazyWalk are sampling-based algorithms and are parallelized using 8 threads. Each experiment is repeated 10 times, and the average results are reported.
\begin{figure*}[htbp!]
	\centering
	\includegraphics[width=1\columnwidth]{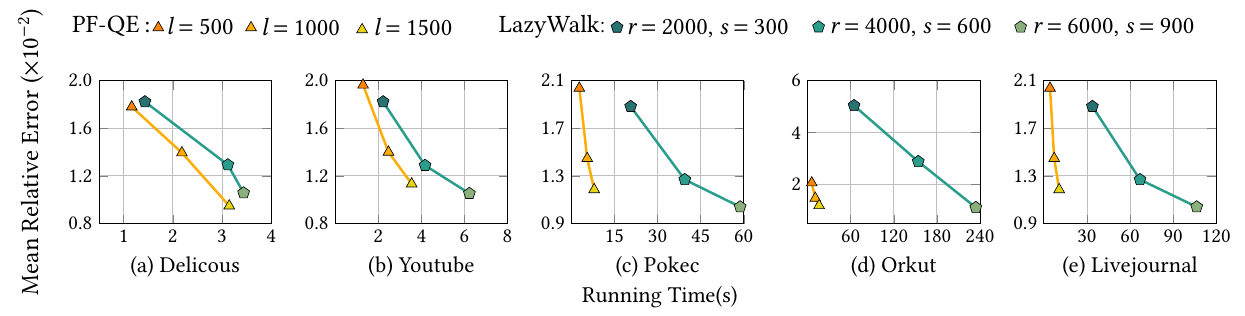}
	\caption{Comparison of mean relative error for $z_i$ and running time under varying parameters for PF-QE and LazyWalk. Internal opinions are generated using the uniform distribution.}\label{f1}	
\end{figure*}

From Figure~\ref{f1}, we observe that the PF-QE curves consistently lie below and to the left of those for LazyWalk, indicating that PF-QE achieves better accuracy and lower running time. Based on this observation, we fix $l = 1000$ for PF-QE, and for LazyWalk, we use 4,000 walks with 600 steps, following the configuration in~\cite{NeDoPe24}. For other quantities, we maintain 8-thread parallelization, repeat each experiment 10 times, and present the averaged results in Table~\ref{FJ-estimation}.

\begin{table*}[htbp!]
\small
\renewcommand{\arraystretch}{1.1}
\setlength{\tabcolsep}{2.5pt} 
\caption{Mean relative errors and running time  for three algorithms on six networks. Internal opinions are generated using the uniform distribution.}\label{FJ-estimation}
\begin{tabular}{clccclllllllllllllllllllll}
\cline{1-17}
\multirow{3}{*}{Network} &  & \multicolumn{3}{c}{Time(s)} &  & \multicolumn{11}{c}{Reletive Error in $\%$} &  &  &  &  &  &  &  &  &  \\ \cline{3-5} \cline{7-17}
 &  & \multirow{2}{*}{LapSolver} & \multirow{2}{*}{LazyWalk} & \multirow{2}{*}{PF-QE} &  & \multicolumn{5}{c}{LazyWalk} &  & \multicolumn{5}{c}{PF-QE} & \multicolumn{1}{c}{} & \multicolumn{1}{c}{} & \multicolumn{1}{c}{} & \multicolumn{1}{c}{} &  &  &  &  &  \\ \cline{7-11} \cline{13-17}
 &  &  &  &  &  & \multicolumn{1}{c}{$P$} & \multicolumn{1}{c}{$I$} & \multicolumn{1}{c}{$C$} & \multicolumn{1}{c}{$DC$} & \multicolumn{1}{c}{$D$} &  & \multicolumn{1}{c}{$P$} & \multicolumn{1}{c}{$I$} & \multicolumn{1}{c}{$C$} & \multicolumn{1}{c}{$DC$} & \multicolumn{1}{c}{$D$} & \multicolumn{5}{c}{} &  &  &  &  \\ \cline{1-1} \cline{3-5} \cline{7-11} \cline{13-17}
Delicious &  & 7.2 & 3.1 & 2.2 &  & 0.90 & 0.48 & 0.83 & 0.42 & 3.58 &  & 1.38 & 0.42 & 0.55 & 0.32 & 4.05 & \multicolumn{1}{c}{} & \multicolumn{1}{c}{} & \multicolumn{1}{c}{} & \multicolumn{1}{c}{} & \multicolumn{1}{c}{} &  &  &  &  \\
Youtube &  & 9.6 & 4.1 & 2.4 &  & 1.03 & 0.75 & 0.21 & 0.33 & 3.52 &  & 0.95 & 0.73 & 0.22 & 0.31 & 4.19 &  &  &  &  &  &  &  &  &  \\
Pokec &  & 49.8 & 39.5 & 5.3 &  & 1.83 & 0.93 & 0.25 & 0.50 & 11.01 &  & 2.01 & 0.96 & 0.21 & 0.51 & 10.63 &  &  &  &  &  &  &  &  &  \\
Orkut &  & 291.9 & 142.2 & 9.9 &  & 3.97 & 0.62 & 5.43 & 2.71 & 34.77 &  & 2.98 & 0.72 & 0.23 & 0.47 & 13.67 &  &  &  &  &  &  &  &  &  \\
Livejournal &  & 186.8 & 67.6 & 6.9 &  & 1.37 & 0.85 & 0.31 & 0.61 & 8.36 &  & 1.30 & 0.90 & 0.30 & 0.52 & 8.17 &  &  &  &  &  &  &  &  &  \\
Twitter &  & - & 160.6 & 32.3 &  & \multicolumn{1}{c}{-} & \multicolumn{1}{c}{-} & \multicolumn{1}{c}{-} & \multicolumn{1}{c}{-} & \multicolumn{1}{c}{-} &  & \multicolumn{1}{c}{-} & \multicolumn{1}{c}{-} & \multicolumn{1}{c}{-} & \multicolumn{1}{c}{-} & \multicolumn{1}{c}{-} &  &  &  &  &  &  &  &  &  \\ \cline{1-17}
\end{tabular}
\end{table*}


The results show that our algorithms achieve relative errors below $1.5\%$ for average expressed opinion and below $4\%$ for $P$, $I$, $C$, and $DC$. Errors for $D$ are higher due to its dependence on both $DC$ and $C$, amplifying estimation errors. PF runs faster than both LapSolver and LazyWalk while maintaining competitive accuracy. For the Twitter network, LapSolver cannot process the data due to time and memory constraints, whereas our algorithm handles it efficiently, running approximately five times faster than LazyWalk.

\subsection{  Optimization Problems}
In this subsection, we evaluate the performance of our proposed algorithms for solving the FJ optimization problems OpMin and PDMin on real-world networks.  We set   $l = 1000$, $p = 10000$, and $k=50$ for our algorithms. Since our algorithms are sampling-based, we repeat each experiment 10 times and report the average results to ensure robustness and consistency.

\begin{table}[htbp!]
\small
\centering
\renewcommand{\arraystretch}{1.1}
\setlength{\tabcolsep}{5pt}
\caption{Running time and effectiveness of \textsc{OpMin} (in terms of opinion decline $\theta$) and \textsc{PDMin} (in terms of P-D index decline $\delta$).}\label{FJ-estimation-shx}
\begin{tabular}{cccccccccllccccc}
\hline
\multirow{3}{*}{Network} & \multicolumn{1}{l}{} & \multicolumn{8}{c}{OpMin} &  & \multicolumn{5}{c}{PDMin} \\ \cline{3-10} \cline{12-16} 
 &  & \multicolumn{2}{c}{\textsc{PF-OpMin}} &  & \multicolumn{2}{c}{\textsc{Fast}} &  & \multicolumn{2}{c}{\textsc{Exact}} &  & \multicolumn{2}{c}{\textsc{PF-PDMin}} &  & \multicolumn{2}{c}{\textsc{FastGreedy}} \\ \cline{3-4} \cline{6-7} \cline{9-10} \cline{12-13} \cline{15-16} 
 &  & time & $\theta$ &  & time & $\theta$ &  & time & \multicolumn{1}{c}{$\theta$} &  & time & $\delta$ &  & time & $\delta$ \\ \cline{1-1} \cline{3-10} \cline{12-16} 
Delicious &  & 2.2 & -6.21 &  & 12.4 & -6.79 &  & 5.9 & -7.80 &  & 160.3 & -5.69 &  & 862.3 & -5.63 \\
Youtube &  & 3.7 & -6.56 &  & 33.2 & -6.48 &  & 4.6 & -7.79 &  & 154.3 & -6.46 &  & 1814.3 & -6.42 \\
Pokec &  & 6.3 & -4.17 &  & 79.4 & -4.13 &  & 60.2 & -4.98 &  & 494.1 & -3.18 &  & 18656 & -3.17 \\
Orkut &  & 10.3 & -1.88 &  & 163.1 & -1.79 &  & 310.5 & -2.43 &  & 774.2 & -5.23 &  & - & - \\
Livejournal &  & 8.1 & -6.56 &  & 296.8 & -6.44 &  & 214.7 & -7.12 &  & 436.5 & -4.62 &  & - & - \\
Twitter &  & 33.7 & - &  & 625.0 & - &  & - &  &  & 2531 & - &  & - & - \\ \hline
\end{tabular}
\end{table}

For OpMin, we compare our algorithm with the \textsc{Fast} forest sampling method \cite{SuZh23}. To ensure a fair comparison, we set the number of forests in \cite{SuZh23} to 1000, consistent with our experimental setup, and execute both algorithms in parallel using 8 threads. The \textsc{Exact} algorithm, which uses the Laplacian solver, serves as the baseline method.  Table~\ref{FJ-estimation-shx} reports the running time and the average expressed opinion decline $\theta$ after setting the internal opinions of 50 nodes to zero, comparing three methods: \textsc{PF-OpMin}, \textsc{Fast}, and \textsc{Exact}. The results show that our proposed algorithm, \textsc{PF-OpMin}, achieves comparable or even better effectiveness than \textsc{Fast}, while being significantly more efficient. For instance, on the LiveJournal network, \textsc{PF-OpMin} is approximately 37 times faster than \textsc{Fast}, while also achieving a slightly larger opinion decline.

For the PDMin problem, we compare our proposed algorithm \textsc{PF-PDMin} with the \textsc{FastGreedy} method~\cite{ZhBaZh21}, which utilizes the Johnson–Lindenstrauss Lemma~\cite{JoLi84} and a fast Laplacian solver~\cite{SpTe14}. In our experiments, we follow the parameter setting in~\cite{ZhBaZh21}, setting the size of the candidate edge set $E_C$ to $10^4$. And we use 20 iterations for the JL lemma.  For our method, we fix $k = 50$ and $l = 1000$. Table~\ref{FJ-estimation-shx} summarizes the running time and the reduction in polarization and disagreement, denoted by $\delta$, for both methods. The results demonstrate that \textsc{PF-PDMin} is substantially more efficient and consistently achieves better performance than \textsc{FastGreedy}. For instance, on the Pokec network, our algorithm achieves more than a 35$\times$ speedup while also yielding a greater decline in the polarization-disagreement index. Furthermore, \textsc{FastGreedy} fails to scale to the three largest networks due to time and memory constraints, whereas our method remains both efficient and scalable.

\section{Conclusions}
In this paper, we proposed several algorithms to compute and optimize opinion-based metrics for the Friedkin-Johnsen (FJ) model. We first introduced the concept of partial rooted forests and presented an efficient algorithm for computing several opinion-based quantities.      Our proposed algorithms, \textsc{PF-OpMin} and \textsc{PF-PDMin}, reduce  the time complexity from linear to sublinear compared to state-of-the-art methods.   Extensive experiments on real-world networks demonstrate  that our algorithms are both accurate and efficient, outperforming state-of-the-art methods and scaling effectively to large networks with over 20 million nodes.

Currently, our framework has limitations in handling optimization problems where the objective involves nonlinear terms such as squared components (e.g., $z_i^2$ in the marginal gain). In future work, we plan to extend our methods to support a broader class of opinion optimization problems under the FJ model.

\section*{Acknowledgements}
This work was supported by the National Natural Science Foundation of China (Nos. 62372112).

\bibliographystyle{plain} 
\bibliography{reference}

 \clearpage

\newpage
\section*{NeurIPS Paper Checklist}

\begin{enumerate}

\item {\bf Claims}
    \item[] Question: Do the main claims made in the abstract and introduction accurately reflect the paper's contributions and scope?
    \item[] Answer:  \answerYes{}
    \item[] Justification: The abstract and introduction clearly state the contributions and scope of our proposed algorithm, including both theoretical and experimental results.

\item {\bf Limitations}
    \item[] Question: Does the paper discuss the limitations of the work performed by the authors?
    \item[] Answer: \answerYes{}
    \item[] Justification: We  discuss the limitations of the work in Conclusions.

\item {\bf Theory assumptions and proofs}
    \item[] Question: For each theoretical result, does the paper provide the full set of assumptions and a complete (and correct) proof?
    \item[] Answer: \answerYes{}
    \item[] Justification: All theoretical results have clearly stated assumptions, and complete proofs are provided in the appendix.

    \item {\bf Experimental result reproducibility}
    \item[] Question: Does the paper fully disclose all the information needed to reproduce the main experimental results of the paper to the extent that it affects the main claims and/or conclusions of the paper (regardless of whether the code and data are provided or not)?
    \item[] Answer: \answerYes{}
    \item[] Justification: All information required to reproduce the main results is provided in Section Experiments.

\item {\bf Open access to data and code}
    \item[] Question: Does the paper provide open access to the data and code, with sufficient instructions to faithfully reproduce the main experimental results, as described in supplemental material?
    \item[] Answer: \answerYes{}
    \item[] Justification: We provide open access to the code at the anonymous link: \url{https://github.com/HaoxinSun98/FJ-PF}. This URL is included in the submission, and the data used in our experiments are all publicly available.

\item {\bf Experimental setting/details}
    \item[] Question: Does the paper specify all the training and test details (e.g., data splits, hyperparameters, how they were chosen, type of optimizer, etc.) necessary to understand the results?
    \item[] Answer:  \answerYes{}
    \item[] Justification: The full details can be found in Section Experiments.

\item {\bf Experiment statistical significance}
    \item[] Question: Does the paper report error bars suitably and correctly defined or other appropriate information about the statistical significance of the experiments?
\item[] Answer: \answerYes{} 
\item[] Justification: We show the standard deviation in the APPENDIX.

\item {\bf Experiments compute resources}
    \item[] Question: For each experiment, does the paper provide sufficient information on the computer resources (type of compute workers, memory, time of execution) needed to reproduce the experiments?
    \item[] Answer: \answerYes{}
    \item[] Justification: The details can be found in Section Experiments.
    
\item {\bf Code of ethics}
    \item[] Question: Does the research conducted in the paper conform, in every respect, with the NeurIPS Code of Ethics \url{https://neurips.cc/public/EthicsGuidelines}?
    \item[] Answer: \answerYes{} 
    \item[] Justification: This work fully complies with the NeurIPS Code of Ethics. All datasets are publicly available with proper attribution, and the code is released via an anonymous repository. No human subjects or sensitive data were involved.

\item {\bf Broader impacts}
    \item[] Question: Does the paper discuss both potential positive societal impacts and negative societal impacts of the work performed?
    \item[] Answer: \answerNo{} 
    \item[] Justification: This paper does not explicitly discuss societal impacts. However, the problems addressed in our work are based on previously published formulations~\cite{ZhBaZh21} presented at NeurIPS 2021, and our focus is on improving algorithmic efficiency.
    
\item {\bf Safeguards}
    \item[] Question: Does the paper describe safeguards that have been put in place for responsible release of data or models that have a high risk for misuse (e.g., pretrained language models, image generators, or scraped datasets)?
    \item[] Answer: \answerNA{} 
    \item[] Justification: This paper does not involve the release of high-risk data or models such as pretrained language models.  

\item {\bf Licenses for existing assets}
    \item[] Question: Are the creators or original owners of assets (e.g., code, data, models), used in the paper, properly credited and are the license and terms of use explicitly mentioned and properly respected?
    \item[] Answer: \answerYes{} 
    \item[] Justification: All existing assets (datasets and code) are publicly available with proper citations in Section Experiments.

\item {\bf New assets}
    \item[] Question: Are new assets introduced in the paper well documented and is the documentation provided alongside the assets?
    \item[] Answer: \answerYes{} 
    \item[] Justification:  The source code is released at \url{https://github.com/HaoxinSun98/FJ-PF}.

\item {\bf Crowdsourcing and research with human subjects}
    \item[] Question: For crowdsourcing experiments and research with human subjects, does the paper include the full text of instructions given to participants and screenshots, if applicable, as well as details about compensation (if any)? 
    \item[] Answer: \answerNA{} 
    \item[] Justification: This work does not involve any crowdsourcing, human subject research. All experiments are conducted on publicly available datasets.  

\item {\bf Institutional review board (IRB) approvals or equivalent for research with human subjects}
    \item[] Question: Does the paper describe potential risks incurred by study participants, whether such risks were disclosed to the subjects, and whether Institutional Review Board (IRB) approvals (or an equivalent approval/review based on the requirements of your country or institution) were obtained?
    \item[] Answer: \answerNA{} 
    \item[] Justification: This work does not involve any human subject research, personal data collection, or activities requiring IRB approval.

\item {\bf Declaration of LLM usage}
    \item[] Question: Does the paper describe the usage of LLMs if it is an important, original, or non-standard component of the core methods in this research? Note that if the LLM is used only for writing, editing, or formatting purposes and does not impact the core methodology, scientific rigorousness, or originality of the research, declaration is not required.
    \item[] Answer: \answerNA{} 
    \item[] Justification: The core methodology of this work   does not involve any use of large language models (LLMs) as an original or non-standard component. LLMs were not employed in algorithm design, experiments, or analysis.

\end{enumerate}

\newpage
\appendix
\section{APPENDIX}




\subsection{Pseudocode of Algorithm \textsc{PFS}}
 
\begin{algorithm}[htbp!]
\caption{$\textsc{PFS}(\calG,\SS,l)$}
\label{alg:partial-rooted-forest}
\Input{Graph $\calG=(V,E)$, node set $\SS = \{v_1, \ldots, v_p\}$, number of samples $l$}
\Output{Partial rooted forest list $L$}
$L \leftarrow [\ ]$\;
\For{$k \leftarrow 1$ \KwTo $l$}{
	$V_\phi \leftarrow \emptyset$, $E_\phi \leftarrow \emptyset$\;
	\For{$i \leftarrow 1$ \KwTo $p$}{
		$u \leftarrow v_i$, $P_u \leftarrow [\ ]$, $c \leftarrow u$\;
		\While{True}{
			\eIf{$c \in V_\phi$}{
				\textbf{break}\;
			}{
				\eIf{random() < $\frac{1}{1 + d_{c}}$}{
					$P_u \leftarrow P_u \cup \{c\}$\;
					mark $c$ as root node\;
					\textbf{break}\;
				}{
					$next \leftarrow$ random choice from $N^+_{c}$\;
					$P_u \leftarrow P_u \cup \{c, (c, next)\}$\;
					$c \leftarrow next$\;
				}
			}
		}
		$\hat{P}_u \leftarrow$ Apply loop-erasure to $P_u$\;
		$V_\phi \leftarrow V_\phi \cup \text{nodes in } \hat{P}_u$\;
		$E_\phi \leftarrow E_\phi \cup \text{edges in } \hat{P}_u$\;
	}
	Add $\phi = (V_\phi, E_\phi)$ to $L$\;
}
\Return $L$\;
\end{algorithm}   

\subsection{Proof of Lemma~\ref{le-estimate}}
\begin{proof}
	Since $z_i = \sum_{j=1}^n \omega_{ij}s_j$, the proof is reduced to showing  $\mathbb{E}[\mathbb{I}_{\{r_{\phi}(i)=j\}}] = \omega_{ij}$. Consider the proof of Theorem~\ref{th-partial-rooted-forest}, and we extend the partial rooted forest $\phi$ to a complete rooted forest $\phi'$. In $\phi$ and $\phi'$, we have $r_{\phi}(i) = r_{\phi'}(i)$ directly, since the nodes in $\phi$ are a subset of the nodes in $\phi'$. Using the Matrix Forest Theorem~\cite{ChSh97,ChSh98}, we have $\mathbb{E}[\mathbb{I}_{\{r_{\phi'}(i)=j\}}] = \omega_{ij}$. Thus, $\mathbb{E}[\mathbb{I}_{\{r_{\phi}(i)=j\}}] = \omega_{ij}$, and the estimator $\widehat{z}_i$ is unbiased.
\end{proof}

\subsection{Hoeffding's inequality}

 \begin{lemma}[Hoeffding's inequality~\cite{Ho94}]\label{le-ho}
Let $x_1,x_2,\cdots, x_n$ be $l$ independent random variables satisfying $a \leq x_i \leq b$ for all $i=1,2,\cdots,n$. Let $x=\frac{1}{l}\sum_{i=1}^l x_i$. Then for any $\epsilon>0$, $\mathbb{P}(|x-\mathbb{E}(x)| \ge \epsilon) \le 2 \, {\rm exp}\left(-\frac{2l \epsilon^2}{(b-a)^2}\right)$. 
\end{lemma}

\subsection{Proof of Theorem~\ref{th-estimation}}
\begin{proof}
Hoeffding's inequality, presented in Lemma~\ref{le-ho}, serves as a powerful tool for estimating the required sample size $l$ to achieve a desired error guarantee. This inequality has been widely applied in various studies, such as~\cite{SuZh23, NeDoPe24}. Utilizing Lemma~\ref{le-estimate} alongside Hoeffding's inequality, selecting $l   = O(\frac{1}{\epsilon^2}\log\frac{1}{\delta})$ ensures that the estimated value $\widehat{z_i}$ deviates from the true value $z_i$ by at most an absolute error of $\epsilon$, with probability at least $1 - \delta$, for any node $i\in V$. Furthermore, according to Lemma 2.2 in Section 6 in~\cite{BhArYo22}, by choosing $p =   O(\frac{1}{\epsilon^2}\log\frac{1}{\delta})$ and setting $\SS = \{v_1, \cdots, v_p\}$, we can guarantee that the sum $\sum_{i \in \SS} {z_i}$ approximates $\sum_{i \in  V} z_i$ within an absolute error of $n\epsilon$, with probability at least $1 - \delta$. It is important to note that $z_i \in [0,1]$, by setting $l =   O(\frac{1}{\epsilon^2}\log\frac{1}{\delta})$ and $p = O(\frac{1}{\epsilon^2}\log\frac{1}{\delta})$, the opinion-based measures $D(\calG)$, $P(\calG)$, $I(\calG)$, $C(\calG)$, and $DC(\calG)$ can be estimated within an absolute error of $n\epsilon$, with probability at least $1 - \delta$.
\end{proof}

\subsection{Further Experiment Results}
In this subsection, we report the standard deviations for two algorithms across five networks in Table~\ref{tb5} and provide further experimental results based on the exponential distribution in Table~\ref{tb6} and Table~\ref{tb7}. 

\begin{table*}[htbp!]
\setlength{\tabcolsep}{1.5 mm}
\renewcommand{\arraystretch}{1.1}
\caption{Standard Deviation for two algorithms on five networks. Internal opinions are generated using the uniform distribution.}\label{tb5}
	\begin{tabular}{ccccccccccccccc}
\hline
\multirow{3}{*}{Network} &  & \multicolumn{13}{c}{Standard Deviation in $\%$} \\ \cline{3-15} 
 &  & \multicolumn{6}{c}{LazyWalk} & \multicolumn{1}{l}{} & \multicolumn{6}{c}{PF-QE} \\ \cline{3-8} \cline{10-15} 
 &  & $z$ & $P$ & $I$ & $C$ & $DC$ & $D$ &  & $z$ & $P$ & $I$ & $C$ & $DC$ & $D$ \\ \cline{1-1} \cline{3-8} \cline{10-15} 
Delicious &  & 0.12 & 6.33 & 6.94 & 2.73 & 4.15 & 3.59 &  & 0.14 & 7.27 & 7.02 & 2.86 & 4.20 & 3.67 \\
Youtube &  & 0.006 & 0.93 & 0.70 & 0.13 & 0.27 & 2.18 &  & 0.009 & 0.98 & 0.73 & 0.17 & 0.30 & 2.04 \\
Pokec &  & 0.006 & 1.91 & 0.43 & 0.10 & 0.15 & 4.83 &  & 0.01 & 1.88 & 0.43 & 0.08 & 0.12 & 5.06 \\
Orkut &  & 0.02 & 4.02 & 0.91 & 0.16 & 0.71 & 44.83 &  & 0.02 & 3.32 & 0.92 & 0.25 & 0.47 & 24.17 \\
Livejournal &  & 0.006 & 1.15 & 0.55 & 0.14 & 0.33 & 4.84 &  & 0.05 & 1.17 & 0.59 & 0.43 & 0.36 & 4.60 \\ \hline
\end{tabular}
\end{table*}

\begin{table*}[htbp!]
\small
\renewcommand{\arraystretch}{1.1}
\setlength{\tabcolsep}{1.8pt} 
\caption{Mean relative errors and running times for three algorithms on six networks. Internal opinions are generated using the exponential distribution.}\label{tb6}
\begin{tabular}{clccccccccccccccccc}
\hline
\multirow{3}{*}{Network} &  & \multicolumn{3}{c}{Time(s)} &  & \multicolumn{13}{c}{Reletive Error in $\%$} \\ \cline{3-5} \cline{7-19} 
 &  & \multirow{2}{*}{LapSolver} & \multirow{2}{*}{LazyWalk} & \multirow{2}{*}{PF-QE} &  &  & \multicolumn{5}{c}{LazyWalk} &  & \multicolumn{6}{c}{PF-QE} \\ \cline{7-12} \cline{14-19} 
 &  &  &  &  &  & $z$ & \multicolumn{1}{c}{$P$} & \multicolumn{1}{c}{$I$} & \multicolumn{1}{c}{$C$} & \multicolumn{1}{c}{$DC$} & \multicolumn{1}{c}{$D$} &  & $z$ & \multicolumn{1}{c}{$P$} & \multicolumn{1}{c}{$I$} & \multicolumn{1}{c}{$C$} & \multicolumn{1}{c}{$DC$} & \multicolumn{1}{c}{$D$} \\ \cline{1-1} \cline{3-5} \cline{7-12} \cline{14-19} 
Delicious &  & 6.6 & 2.3 & 2.1 &  & 2.50 & 2.79 & 3.09 & 1.66 & 1.17 & 3.81 &  & 2.41 & 1.96 & 3.05 & 0.79 & 1.10 & 3.75 \\
Youtube &  & 9.6 & 4.4 & 2.6 &  & 1.90 & 3.25 & 2.93 & 0.78 & 1.41 & 6.32 &  & 2.36 & 3.33 & 2.89 & 0.77 & 1.38 & 6.49 \\
Pokec &  & 48.5 & 42.7 & 5.4 &  & 3.35 & 3.72 & 2.89 & 0.48 & 1.32 & 10.91 &  & 2.50 & 3.64 & 2.72 & 0.43 & 1.30 & 10.85 \\
Orkut &  & 320.8 & 153.2 & 9.7 &  &2.92  & 8.00 & 1.95 & 5.29 & 2.06 & 76.66 &  &2.51  & 12.54 & 1.97 & 0.41 & 1.02 & 25.67 \\
Livejournal &  & 239.5 & 70.7 & 6.6 &  &2.44  &3.22 & 2.87 & 1.82 & 1.24 & 5.06 &  &2.40  & 3.10 & 2.89 & 1.38 & 1.50 & 5.11 \\
Twitter &  & - & 168.6 & 34.9 &  &  & \multicolumn{1}{c}{-} & \multicolumn{1}{c}{-} & \multicolumn{1}{c}{-} & \multicolumn{1}{c}{-} & \multicolumn{1}{c}{-} &  &  & \multicolumn{1}{c}{-} & \multicolumn{1}{c}{-} & \multicolumn{1}{c}{-} & \multicolumn{1}{c}{-} & \multicolumn{1}{c}{-} \\ \hline
\end{tabular}
\end{table*}

\begin{table*}[htbp!]
\setlength{\tabcolsep}{1.5 mm}
\renewcommand{\arraystretch}{1.1}
\caption{Standard Deviation for two algorithms on five networks. Internal opinions are generated using the exponential distribution.}\label{tb7}
\begin{tabular}{ccccccccccccccc}
\hline
\multirow{3}{*}{Network} &  & \multicolumn{13}{c}{Standard Deviation in $\%$} \\ \cline{3-15} 
 &  & \multicolumn{6}{c}{LazyWalk} & \multicolumn{1}{l}{} & \multicolumn{6}{c}{PF-QE} \\ \cline{3-8} \cline{10-15} 
 &  & $z$ & $P$ & $I$ & $C$ & $DC$ & $D$ &  & $z$ & $P$ & $I$ & $C$ & $DC$ & $D$ \\ \cline{1-1} \cline{3-8} \cline{10-15} 
Delicious &  & 0.008 & 1.49 & 2.15 & 0.59 & 0.87 & 3.07 &  & 0.01 & 1.37 & 2.12 & 0.57 & 0.89 & 3.00 \\
Youtube &  & 0.01 & 2.06 & 2.36 & 1.02 & 1.58 & 4.58 &  & 0.02 & 1.97 & 2.32 & 0.99 & 1.52 & 4.50 \\
Pokec &  & 0.01 & 3.29 & 1.95 & 0.33 & 0.87 & 8.12 &  & 0.02 & 3.24 & 1.87 & 0.34 & 0.883 & 7.63 \\
Orkut &  & 0.02 & 6.37 & 1.18 & 0.30 & 1.31 & 27.45 &  & 0.02 & 7.66 & 1.10 & 0.26 & 0.84 & 21.37 \\
Livejournal &  & 0.009 & 2.86 & 2.38 & 0.81 & 1.30 & 5.25 &  & 0.05 & 2.88 & 2.38 & 1.02 & 1.29 & 3.95 \\ \hline
\end{tabular}
\end{table*}

\end{document}